\documentclass[final,twoside,11pt]{entics} 
\usepackage{cmll}
\usepackage{ebproof}
\usepackage{enticsmacro}
\usepackage{graphicx}
\usepackage{pn}
\usepackage{tikz}
\usetikzlibrary{positioning}
\usepackage{tikz-cd}
\usepackage
    {todonotes}
\usepackage{subcaption}
\usepackage{wrapfig}
\newcommand{\beforepn}{-3\baselineskip}
\newcommand{\MELL}{\mathsf{MELL}}
\newcommand{\MALL}{\mathsf{MALL}}
\newcommand{\MLL}{\mathsf{MLL}}
\newcommand{\Mix}{\mathsf{Mix}}
\newcommand{\MLLJ}{\MLL_{\J}}
\newcommand{\PS}{\mathsf{PS}}
\newcommand{\PSJ}{\PS_{\J}}
\newcommand{\J}{\mathbf{J}}
\newcommand{\axj}{A_\mathbf{J}}
\newcommand{\erasure}[2]{\left.#1\right|_{#2}}
\newcommand{\erase}[2]{\mathsf{erase}_{#2}\left(#1\right)}
\newcommand{\erasurej}[1]{\erasure{#1}{\J}}
\newcommand{\erasej}[1]{\erase{#1}{\J}}
\renewcommand{\L}{\mathcal{L}}
\newcommand{\F}{\mathcal{F}}
\newcommand{\seq}{\mathsf{seq}}
\newcommand{\ctr}{\mathsf{ctr}}
\newcommand{\cut}{\mathsf{cut}}
\newcommand{\ax}{\mathsf{ax}}
\newcommand{\One}{\mathbf{1}}

\def\ctodo#1{}
\usepackage[
    type={CC},
    modifier={by},
    version={3.0},
]{doclicense}

\def\conf{MFPS 2023} 	
\def\myconf{mfps39}
\volume{3}			
\def\volu{3}			
\def\papno{2}			
\def\lastname{Alcolei, Pellissier, Saurin} 
\def\runtitle{The exponential logic of sequentialization} 
\def\mydoi{12419}
\def\copynames{A. Alcolei, L. Pellissier, A. Saurin}    

\usepackage{amsmath}
\begin{document}
\begin{frontmatter}

  \title{The Exponential Logic of Sequentialization}
  \author{Aurore Alcolei\thanksref{a}\thanksref{aemail}}
  \author{Luc Pellissier\thanksref{b}\thanksref{bemail}}
  \author{Alexis Saurin\thanksref{c}\thanksref{cemail}}

  \address[a]{Université Paris Est Creteil, LACL, F-94010 Créteil, France}
  \thanks[aemail]{aurore.alcolei@ens-lyon.org}
  \address[b]{Université Paris Est Creteil, LACL, F-94010 Créteil, France}
  \thanks[bemail]{luc.pellissier@u-pec.fr}
  \address[c]{IRIF, CNRS, Université Paris Cité \& INRIA, Paris, France}
  \thanks[cemail]{alexis.saurin@irif.fr}

  \begin{abstract}

    Linear logic has provided new perspectives on proof-theory, denotational
    semantics and the study of programming languages. One of its main successes
    are proof-nets, canonical representations of proofs that lie at the
    intersection between logic and graph theory. In the case of the minimalist
    proof-system of multiplicative linear logic without units (MLL), these two
    aspects are completely fused: proof-nets for this system are graphs satisfying
    a correctness criterion that can be fully expressed in the language of graphs.

    For more expressive logical systems (containing logical constants, quantifiers
    and exponential modalities), this is not completely the case. The purely
    graphical approach of proof-nets deprives them of any sequential structure
    that is crucial to represent the order in which arguments are presented, which
    is necessary for these extensions. Rebuilding this order of presentation —
    sequentializing the graph — is thus a requirement for a graph to be logical.
    Presentations and study of the artifacts ensuring that sequentialization can
    be done, such as boxes or jumps, are an integral part of researches on linear
    logic.

    Jumps, extensively studied by Faggian and di Giamberardino, can express
    intermediate degrees of sequentialization between a sequent calculus proof and
    a fully desequentialized proof-net. We propose to analyze the logical strength
    of jumps by internalizing them in an extention of MLL where axioms on a specific
    formula, the jumping formula, introduce constrains on the possible sequentializations. The jumping
    formula needs to be treated non-linearly, which we do either axiomatically, or
    by embedding it in a very controlled fragment of multiplicative-exponential
    linear logic, uncovering the exponential logic of sequentialization.
  \end{abstract}
\end{frontmatter}

\section{Introduction}
\label{sec:introduction}

Proof theory is concerned with proof systems, specifying how proofs are
structured, and in turn, which formul\ae{} are provable. Different proof systems
have different properties, such as
soundness (the fact that all provable
formul\ae{} are also valid —~for a given semantics of formul\ae{}),
completeness
(the fact that all valid formul\ae{} are also provable —~again, for a given semantics of formul\ae{}),
canonicity (two proofs differ if and only if they have different
interpretations —~for a given semantics of proofs), ...

Broadly speaking, Linear logic~\cite{Girard87} has two main kinds of proof systems:
\begin{description}
\item[Sequent calculus] Proofs are trees whose nodes are labelled by inference
  rules transforming sequents of formul\ae{}. This representation --~even though
  it is both historically decisive in the development of proof-theory since the
  work of Gentzen by emphasizing the role of the cut inference as well as for
  reductive logics (aka. proof-construction) by suggesting a natural notion of
  proof-goal~-- contains lot of unnecessary information (such as the specific
  order of application of some unrelated rules): it is far from being canonical.
  The loss of confluence in sequent calculus witnesses this loss of canonicity.
\item[Proof-nets] are graphs whose nodes relate formul\ae{} to one another. A
  proof-net contains much less information than a sequent calculus proof with
  respect to the order of application of inference rules. This allows them to be
  canonical, for some fragments of linear logic at least. In particular, they
  are canonical for multiplicative linear logic without units (simply referred
  to as $\MLL$ in the following). As a by-product, confluence is recovered in
  $\MLL$ proof-nets.
\end{description}

Apart from the contrast on canonicity issue, a major difference between sequent
proofs and proof-nets follows from the tree {\it vs.} graph structure and
concerns logical correctness of the proof-object: sequent proofs being
inductively defined proof objects, their logical correctness is naturally
defined in a \emph{local} way: soundness of the conclusion follows from the
soundness of each premise and one only considers derivation trees that is built
applying inferences in a correct way. On the other hand, proof-nets being
graphs, there is no such thing as a root, in general, and logical correctness
therefore becomes a \emph{global} property of the graph.  
One of this property is \emph{saturation}, the fact that the graph can
be saturated with more edges to produce a proof tree, called a \emph{sequentialization}
of the net. \ctodo{ utile ? cf comments sur l'introduction du terme sequentialization dans l'intro }
It would be ad hoc to
restrict solely to those graphs which are logically correct and one therefore
consider two levels of proof objects: \textbf{proof structures}, which are
graphs with some typing constraints on the edges and vertices; and
\textbf{proof-nets}, which are logically correct proof structures.

For instance, consider the following rules of linear logic sequent calculus:
\[
  \begin{prooftree}
    \infer0[$\One$]{\vdash \One}
  \end{prooftree}
  \qquad
  \begin{prooftree}
    \hypo{\vdash \Gamma}
    \infer1[$\bot$]{\vdash \Gamma,\bot}
  \end{prooftree}
\]
They introduce the so-called multiplicative constants \(\One\) and
\(\bot\). The way they are introduced differ on their context: while \(\One\)
can be deduced thanks to an axiom, in an empty content, \(\bot\) needs to be juxtaposed to an already
derived context. So, in a proof-net system, where operations are not on sequents
of formul\ae{} but on formul\ae{}, this contextuality is missing. Indeed, both
introductions are presented as a lone node with no input and an output:
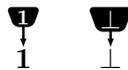
\vspace{\beforepn}
\begin{proofnet}
  \pnformulae{
    \pnf[1]{$\One$}~~\pnf[b]{$\bot$}
  }
  \pnone{1}
  \pnbot{b}
\end{proofnet}
Context has thus to be recovered. One way to do so is via correctness
criterion, an extra \emph{property} verified by the graph,
ensuring that there exists a coherent way of assigning its context
to each formula. Another way is to add extra
\emph{structure} on top of the graph. The most common such structure is the
\emph{box}, introduced in Girard's seminal paper~\cite{Girard87}. Consider the sequent calculus rule introducing the exponential
modality \(\oc\):
\[
  \begin{prooftree}
    \hypo{\vdash \wn \Gamma, A}
    \infer1[$\oc$]{\vdash \wn \Gamma,\oc A}
  \end{prooftree}
\]
For a formula to receive a \(\oc\), all the other formul\ae{} in the context
have to be under the \(\wn\) modality. A way to translate that into proof-nets
is by requiring that a special region of the graph is enclosed so that every
formula going out of it has a \(\wn\), except one corresponding to the
\(\oc\)-cell. This is a direct translation of the sequent calculus (where rules
are applied to sequents and not formul\ae{}) into proof-nets, in a
graph-theoretically awkward fashion:
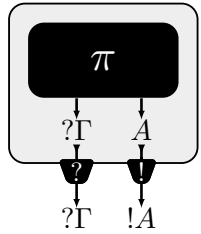
\begin{proofnet}
  \pnsomenet[R]{$\pi$}{2cm}{1cm}
  \pnoutfrom{R.-125}[gamma]{$\wn\Gamma$}
  \pnoutfrom{R.-45}[A]{$A$}
  \pnbox{R,A,gamma}
    \pnprom{A}{$\oc A$}
    \pnauxprom{gamma}{$\wn \Gamma$}
\end{proofnet}

Another more graph-theoretical way to introduce contexts are \emph{jumps}, a
recurrent feature in the literature on proof
nets~\cite{quantifier-In-linear-logic-II}.

They are purely geometrical (i.e. untyped) edges adding extra connection between
links of a proof structure. Jumps were first introduced by Girard
in~\cite{quantifier-In-linear-logic-II} to add more sequentiality in proof nets
with quantifiers, typically preventing sequentialization that would not match
the variable-witness dependencies between $\forall$ and $\exists$ quantifier.
Such jumps can also be found in unification nets, a more recent framework of 
proof nets with quantifiers that bears implicit witnesses~\cite{hughes18-unets}.
Jumps are also used as lighter feature than boxes to handle the sequentiality induced by sequent rules outside of $\MLL$, typically the
unit rules~\cite{lmcs:8871,Girard1996,DBLP:conf/tlca/Laurent99}, additive rules~\cite{Girard1996,DBLP:journals/tocl/HughesG05} and exponentials for
$\lambda$-nets~\cite{DBLP:conf/csl/AccattoliG09}.
\ctodo{(to check and complete,  add a drawing of a jump)}

\begin{figure}
  \vspace{\beforepn}
  \begin{proofnet}
    \pnformulae{
      \pnf[Ab]{$A^{\bot}$}~\pnf[Bb]{$B^{\bot}$}~\pnf[A]{$A$}~\pnf[B]{$B$}\\
      ~~~~~~\pnf[Cb]{$C^{\bot}$}~\pnf[C]{$C$}~~\pnf[Db]{$D^{\bot}$}~\pnf[D]{$D$}
    }
    \pnaxiom{Ab,A}[1][-0.5]
    \pnaxiom{Bb,B}[1][0.5]
    \pnaxiom{Cb,C}
    \pnaxiom{Db,D}
    \pntensor{Ab,Bb}[AboBb]{$A^{\bot}\otimes B^{\bot}$}
    \pnpar{A,B}[ApB]{$A \parr B$}
    \pntensor{ApB,Cb}[ApBoCb]{$(A \parr B)\otimes C^{\bot}$}
    \pntensor{C,Db}[CtDb]{$C \otimes D^{\bot}$}
    \draw[->,dashed,red] (ApBoCb.north east)
    .. controls +(85:3cm) and +(70:1.7cm) .. (C.north east);
 
  \end{proofnet}
  \caption{A graphical jump}
  \label{fig:graphical_jump}
\end{figure}
On another line of work, Faggian and her collaborators developped a 
setting in which jumps are taken seriously as edges in a more general 
graph than just the mere syntactic tree of a formula.
First in the context of
\emph{L-nets}~\cite{strategies-and-proof-nets,curien-faggian-l-nets,faggian-maurel-ludics-net},
a parallel syntax for Ludics designs, then in the context of \emph{J-proof nets}
for $\mathsf{HS}^+$ (a polarized and hypersequentialized sequent calculus) and
$\MLL$~\cite{DiGiamberardinoFaggian:conf,DiGiamberardinoFaggian:journal},
their setting allows in particular to see sequent calculus proofs as proof-nets
saturated with sequential edges, embedding both proof-nets and sequent calculus
proofs in a same universe of partially sequentialized proof-structures, each of
them living at one extremity on a de·sequentialization spectrum.

\ctodo{aurore: pouvez vous relire ce paragraphe ? FAIT (alex) j'ai corrigé quelques typos.}
The figure above depicts a proof net with a jump (the red dashed arrow) from (the formula occurrence) 
$(A\parr B) \otimes C^\perp$ to (the formula occurrence) $C$. 
As a consequence, when sequentializing such a net, the $\otimes$ rule
introducing $(A\parr B)\otimes C^\perp$ will always be placed above
the $\otimes$ rule using $C$ to introduce $C\otimes D^\perp$.
Jumps are thus a way to \emph{constrain} the set of sequentializations 
that are derivable from a proof net.

Proof structures with jumps offer a larger setting than usual proof structures,
it is thus natural to search how the behave under cut elimination and  how the 
usual correctness criteria on proof nets can be adapted to proof nets with jumps. 

%

\paragraph{Contributions and organization of the paper.}
In this paper, we show that graphical jumps as presented above can be
internalised in the logic under study at the only cost of adding 
a special atom $\J$ that enjoys, as well as its dual, the property of
a monoid and behaves as an exponential under cut elimination.
This atom can itself be encoded in a fragment of linear logic, thus giving a 
way to find correctness criteria for proof structure with jumps 
through the ones of usual proof structures.
\ctodo{Aurore: operade ou monoid ?}

To keep this paper light and readable, we will focus our internalisation of jumps to \(\MLL\). In Section~\ref{sec:mll}, we will recast the basic definitions
on $\MLL$ sequent proofs and proof-nets. In Section~\ref{sec:jumps-as-axioms} we start
introducing our encoding of jumps as axioms in a cut-free setting.
 In particular, we will define $\MLLJ$ a logic which is basically
 $\MLL$ plus a special atom $\J$ that enjoys left and right contractions rules.
We will prove that our encoding is correct and that our jumps 
correspond exactly to
the ones described in~\cite{DiGiamberardinoFaggian:journal}, in the sense
that their correctness criterion matches with the usual Danos-Regnier criterion
on our encoding. 
\ctodo{mieux ?}
In Section~\ref{sec:cuts} we will discuss the encoding of jumps in the
dynamic setting of proofs with cuts and cut elimination, enriching
$\MLLJ$ with the appropriate exponential dynamics for $\J$.
This will make explicit the dynamics of jumps sketched
in~\cite{DiGiamberardinoFaggian:journal} and treated more exhaustively in~\cite{giamberardino_2018}. 
Finally, we will conclude in Section~\ref{sec:implem} by showing that
the jump atom $\J$ can be implemented effectively in the
exponential fragment of $\MELL$ + $\Mix$.

\ctodo{Aurore: comment on se situe au niveau de la dynamique par rapport
à ce qui est écrit dans \cite{giamberardino_2018}? il faudrait peut-être en dire plus dans la section 4 et la conclusion}


\section{\texorpdfstring{\(\MLL\)}{MLL} proofs: sequent-calculus, proof-nets and sequentialization}
\label{sec:mll}

\(\MLL\) is a minimal proof-system. Its formul\ae{} are given by the following grammar (with $A$ ranging in a set of atomic formul\ae{}):
\[F = F\otimes F \mid F\parr F \mid A \mid A^\perp.\]

\emph{Linear negation} is inductively defined by
\[ (A^{\perp})^{\perp} =A \qquad
(F\otimes G)^\perp = G^\perp \parr F^\perp \qquad
  (F\parr G)^\perp = G^\perp \otimes F^\perp.
\]
\paragraph{Sequent calculus for $\MLL$}
\(\MLL\) proofs, in the sequent calculus, are given by the applications of the
rules:
\[
  \begin{prooftree}
    \infer0[Ax]{\vdash A, A^{\bot}}
  \end{prooftree}
  \qquad
  \begin{prooftree}
    \hypo{\vdash \Gamma, F, G}
    \infer1[\(\parr\)]{\vdash \Gamma, F\parr G}
  \end{prooftree}
  \qquad
  \begin{prooftree}
    \hypo{\vdash \Gamma, F}
    \hypo{\vdash \Delta, G}
    \infer2[\(\otimes\)]{\vdash \Gamma, \Delta, F\otimes G}
  \end{prooftree}
  \qquad
  \begin{prooftree}
    \hypo{\vdash \Gamma, F}
    \hypo{\vdash \Delta, F^\perp}
    \infer2[$\mathsf{cut}$]{\vdash \Gamma, \Delta}
  \end{prooftree}
\] 
Its dynamic of reduction, also called \emph{cut elimination} is recalled 
in figure~\ref{fig:redMLL}. The resulting calculus is (non deterministic and) not confluent. 

\begin{figure}
  $$\begin{array}{rcl}
  \begin{prooftree}
    \infer0[Ax]{\vdash A, A^{\bot}}
    \hypo{\pi}
    \infer1{\vdash A, \Delta}
    \infer2[$\cut$]{\vdash A, \Delta}
  \end{prooftree}
  & \qquad \to \qquad &
  \begin{prooftree}
    \hypo{\pi}
    \infer1{\vdash A, \Delta}
  \end{prooftree}
  \\\\
  \begin{prooftree}
    \hypo{\theta}
    \infer1{\vdash \Theta, F, G}
    \infer1[\(\parr\)]{\vdash \Theta, F\parr G}
    \hypo{\gamma}
    \infer1{\vdash \Gamma, F^\perp}
    \hypo{\delta}
    \infer1{\vdash \Delta, G^\perp}
    \infer2[\(\otimes\)]{\vdash \Gamma, \Delta, F^\perp \otimes G^\perp}
    \infer2[$\cut$]{\vdash \Theta, \Gamma, \Delta}
  \end{prooftree}
  &\qquad  \to \qquad &
  \begin{prooftree}
    \hypo{\theta}
    \infer1{\vdash \Theta, F, G}
    \hypo{\gamma}
    \infer1{\vdash \Gamma, F^\perp}
    \infer2[$\cut$]{\vdash \Theta,G,  \Gamma}
    \hypo{\delta}
    \infer1{\vdash \Delta, G^\perp}
    \infer2[$\cut$]{\vdash \Theta, \Gamma, \Delta}
  \end{prooftree}
  \\ & \text{or} &
  \begin{prooftree}
    \hypo{\theta}
    \infer1{\vdash \Theta, F, G}
    \hypo{\delta}
    \infer1{\vdash \Delta, G^\perp}
    \infer2[$\cut$]{\vdash \Theta,F,  \Delta}
    \hypo{\gamma}
    \infer1{\vdash \Gamma, F^\perp}
    \infer2[$\cut$]{\vdash \Theta, \Gamma, \Delta}
  \end{prooftree}
\end{array}
  $$
  \caption{Sequent calculus reductions}
  \label{fig:redMLL}
\end{figure}

\begin{remark}
  We have not made precise what sequents are, in the previous definition. There are many possible choices which have various impacts. We leave this as implicit and lightweight as possible in the remaining of the paper but need to discuss this  here. We crucially need a notion of sequents which allows for two properties: (i) to trace \emph{formula occurrences} along the proof, that is admit a well-defined notion of formula ancestor (ii) which validates the exchange rule (as an admissible rule at least).

  Among the standard options at hand, there are sequents as ordered lists of formulas, sequents as sets of locative formula occurrences or finally sequents as indexed sets of formulas. On the other hand, sequents as multisets of formulas are not an option as it does not allow for a well-defined notion of formula occurrence, nor a well-defined notion of desequentialization (see below).
  The first option requires to consider some form of an exchange rule: sequents are lists of formulas, written $\vdash F_1,\dots, F_n$, and the following exchange rule is necessary:
  \(  \begin{prooftree}
    \hypo{\vdash \Gamma, G, F, \Delta}
    \infer1[\(\mathsf{ex}\)]{\vdash \Gamma, F, G, \Delta}
  \end{prooftree}
\)
  .
In such a situation, one can treat, as often, the exchange rule implicitly by considering that the set of inferences is complemented with the derived rules obtained by pre- and post-composing every other rule with a series of exchanges.
(See Example~\ref{ex:ps-two-seq} for an illustration, if sequents are viewed as ordered lists there.) Each such derived rule is equipped with a notion of ancestor derived from the basic inference rules given above (as usually done in the proof theory literature~\cite{buss1998introduction}).

When sequents are treated in a locative way, their formulas are equipped with a unique address (pairwise incomparable) and the notion of sub-address provides an explicit notion of formula ancestor, the exchange rule in unnecessary in this case. (See~\cite{BaeldeDS16,Girard01} for instance.)
%
Last, with sequents as indexed sets of formulas, each inference rules is equipped with a relation between the indexing set of the conclusion of the rule and that of its premisses, constituting the ancestor relation. (Note that the above two presentations are concrete instances of this last presentation: with lists the indexing sets are the positions in the list, with locations, the indexing sets are finite sets of pairwise incomparable addresses.)

In the following, we shall forget about the above considerations and simply assume that we have a well-defined notion of formula occurrence and formula ancestor.
\end{remark}


\paragraph{Proof structures for $\MLL$}
Among sequent calculus proofs of a given sequent, some are essentially different, 
while some are essentially the same: a difference in the order of 
introduction of two connectives is typically a non-essential difference.
A canonical representation of proofs can be found when moving to the highly parallel setting of proof-nets, a graphical representation of proofs.
A \(\MLL\)  \emph{proof-structure} is a directed graph whose nodes (links) and edges abide to the
following typing constraints:\\

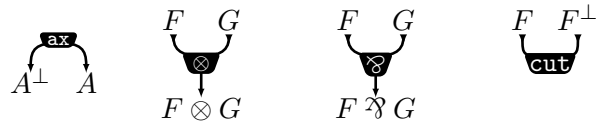
\vspace{\beforepn}
\begin{proofnet}
  \pnformulae{
    ~~~\pnf[F]{$F$}~\pnf[G]{$G$}
    ~~~\pnf[I]{$F$}~\pnf[J]{$G$}
    ~~~\pnf[H]{$F$}~\pnf[Hb]{$F^\bot$}\\
    \pnf[Ab]{$A^{\bot}$}~\pnf[A]{$A$}
  }
  \pnaxiom{Ab,A}
  \pntensor{F,G}{$F\otimes G$}
  \pnpar{I,J}{$F \parr G$}
  \pncut{H,Hb}
\end{proofnet}
Note that the type of the edges of a proof-structure can be recovered from the
type of the edges going out its axiom links and the type of its links, so that 
they could be made implicit although they are often kept for readability reasons.
Besides, we say that such a graph is a \emph{proof-structure for the sequent $\Gamma$} 
if we obtain $\Gamma$ by collecting all the formulas of its pending edges up to
reordering.

Reflecting the dynamics of sequent calculus proofs, proof structures 
also define a (non-deterministic) graph rewriting calculus; the corresponding reductions are recalled in Figure~\ref{fig:proof-structure-red}.
\begin{figure}[t]
  \vspace{\beforepn}
  \centering
  \pnet{
    \pnformulae{
      \pnf[A1]{$A$}~~\pnf[An]{$A^{\bot}$}~~\pnf[A2]{$A$}
    }
    \pnaxiom[ax]{A1,An}
    \pncut[cut]{An,A2}
  }
  \(\longrightarrow\)
  \pnet{
    \pnformulae{
      \pnf[A]{$\quad$}
    }
    \pninto{A.north}[Anorth]{$A$}
    \pnoutfrom{A.south}[Asouth]{$A$}
  }
  \qquad
  \pnet{
    \pnformulae{
      \pnf[A1]{$A$}~~\pnf[B1]{$B$}~~\pnf[B1n]{$B^{\bot}$}~~\pnf[A1n]{$A^{\bot}$}
    }
    \pntensor{A1,B1}[AoB]{$A\otimes B$}
    \pnpar{B1n,A1n}[ApB]{$B^{\bot} \parr A^{\bot}$}
    \pncut[partensor]{AoB,ApB}
  }
  \(\longrightarrow\)
  \pnet{
    \pnformulae{
      \pnf[A2]{$A$}~~\pnf[B2]{$B$}~~\pnf[B2n]{$B^{\bot}$}~~\pnf[A2n]{$A^{\bot}$}
    }
    \pncut[cut1]{A2,A2n}[4]
    \pncut[cut2]{B2,B2n}
  }
  \vspace{\beforepn}
  \caption{Proof-structure reductions}
  \label{fig:proof-structure-red}
\end{figure}
%
Cut elimination in proof-structures is non-deterministic but
\emph{confluent} and \emph{terminating}.
We write $\Downarrow R$ for the (unique) normal form of a proof-structure $R$;
$\Downarrow_\rho R$ if we want to make more precise the sequence of reduction 
steps leading to it\footnote{A sequence of reduction steps is simply given as a sequence of pairs of the activated cut-formulas -- remember that in $\MLL$ proof nets there is no commutation rule.}.

\paragraph{Sequent calculus proofs as (constrained) proof structures}

Let us consider an $\MLL$ proof $\pi$. A formula $F$ is a \emph{formula occurrence} in $\pi$ if it 
is the active conclusion of a rule in $\pi$ (the formula having its topmost connective introduced by the rule or the two formulas introduced by an axiom rule). 
For $F$ a formula occurrence in $\pi$, we define its \emph{life-time}
as sub-branch from its creation to its use.

For two rule occurrences $r$ and $r'$ in a proof $\pi$, we say that $r$ is
\emph{above} (or \emph{occurs before}) $r'$ if $r$ and $r'$ belong 
to the same branch in the proof, and $r'$ appears before $r$ 
starting from the root. 
If $r$ produces a formula $F$ this is equivalent to saying that there 
exists $F'$, an active formula of $r'$, such that $F$ is created before
$F'$ is used; in that case, we write $F \prec F'$.

Every sequent proof $\pi$ can be desequentialized into a proof structure
$R_\pi$, defined inductively. For a proof structure $R$, we note
$\seq(R)$ its set of sequentializations, that is, the set of sequent proofs
$\pi$ such that $R_\pi = R$.
Desequentialization is preserved by cut elimination and, conversely, if
$\pi\in \seq(R)$ and $R\to R'$, then there exists $\pi'\in \seq(R')$ such that
$\pi \to^+ \pi'$. Note, however, that in general proofs in $\seq(R')$ do not correspond
to the reduct of some proof in $\seq(R)$.

If $\pi$ is a sequentialization of $R$ then there is a one-to-one correspondence
between the  rules in $\pi$ and the links in $R$ (this correspondence is defined during
the inductive procedure of desequentialization).
By extension every edge in $R$ can be associated with a formula occurrence in $\pi$
throughout its life-time.
In the following we will often make use of this correspondence, keeping the conversion
from proof to net implicit.

Given a proof structure $R$, its set of sequentializations may be empty: 
in that case we say that $R$ is \emph{incorrect}.
On the contrary, if $\seq(R)$ is non empty, we say that $R$ is \emph{correct}
or that it is \emph{a proof net}, meaning that it actually corresponds to a 
canonical representation of proofs.

A proof net can have several sequentializations, for example:
\begin{example}
  \label{ex:ps-two-seq}
  Consider the (correct) proof-structure \(R\):
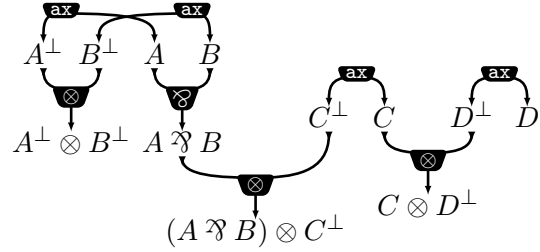
\vspace{\beforepn}
\begin{proofnet}
  \pnformulae{
    \pnf[Ab]{$A^{\bot}$}~\pnf[Bb]{$B^{\bot}$}~\pnf[A]{$A$}~\pnf[B]{$B$}\\
    ~~~~~~\pnf[Cb]{$C^{\bot}$}~\pnf[C]{$C$}~~\pnf[Db]{$D^{\bot}$}~\pnf[D]{$D$}
  }
  \pnaxiom{Ab,A}[1][-0.5]
  \pnaxiom{Bb,B}[1][0.5]
  \pnaxiom{Cb,C}
  \pnaxiom{Db,D}
  \pntensor{Ab,Bb}[AboBb]{$A^{\bot}\otimes B^{\bot}$}
  \pnpar{A,B}[ApB]{$A \parr B$}
  \pntensor{ApB,Cb}{$(A \parr B)\otimes C^{\bot}$}
  \pntensor{C,Db}{$C \otimes D^{\bot}$}
\end{proofnet}
It has two different sequentializations:
\[
  \begin{prooftree}
    \infer0{\vdash A^{\bot},A}
    \infer0{\vdash B^{\bot},B}
    \infer2[\(\otimes\)]{\vdash A^{\bot}\otimes B^{\bot},A,B}
    \infer1[\(\parr\)]{\vdash A^{\bot}\otimes B^{\bot},A\parr B}
    \infer0{\vdash C^{\bot},C}
    \infer2[\(\otimes\)]{\vdash A^{\bot}\otimes B^{\bot}, (A \parr B)\otimes C^{\bot}, C}
    \infer0{\vdash D^{\bot},D}
    \infer2[\(\otimes\)]{\vdash A^{\bot}\otimes B^{\bot}, (A \parr B)\otimes
      C^{\bot}, C \otimes D^{\bot},D}
  \end{prooftree}
  \quad
  \begin{prooftree}
    \infer0{\vdash A^{\bot},A}
    \infer0{\vdash B^{\bot},B}
    \infer2[\(\otimes\)]{\vdash A^{\bot}\otimes B^{\bot},A,B}
    \infer1[\(\parr\)]{\vdash A^{\bot}\otimes B^{\bot},A\parr B}
    \infer0{\vdash C^{\bot},C}
    \infer0{\vdash D^{\bot},D}
    \infer2[\(\otimes\)]{\vdash C^{\bot}, C \otimes D^{\bot},D}
    \infer2[\(\otimes\)]{\vdash A^{\bot}\otimes B^{\bot}, (A \parr B)\otimes
      C^{\bot}, C \otimes D^{\bot},D}
  \end{prooftree}
\]
On the first one, \(A\) is above \(D\), not in the second.
\end{example}

Graphical jumps as depicted in Figure~\ref{fig:graphical_jump} are used to 
restrict the possible
set of sequentializations of a net: if there is a jump from $F$ to $F'$ in $R$, then
we expect $F$ to appear above $F'$ in any sequentialization of $R$.
Writing jumps as a set of pairs of formula occurrences in $R$ and writing $\L$ for a set of jumps, we define
\[\seq_\L(R) := \{\pi \in \seq(R) \mid \forall (F,F')\in \L, F>_\pi F'\}\]
In the example above, the net with jump from $(A\parr B) \otimes C^\perp$ to $C$
has only one sequentialization left: $\seq_{((A\parr B) \otimes C^\perp,C)}(R)$
corresponds to the left sequentialization of Example~\ref{ex:ps-two-seq}.

The aim of this paper is to provide a logical encoding of jumps in proof nets
that validates the definition of $\seq_\L(R)$.
More precisely, given a proof net $R$ and a set of jumps $\L$, we want to define of proof net $\J_\L(R)$ such that $\seq(R) = \seq(\J_\L(R))$.

\section{Jumps as axioms}
\label{sec:jumps-as-axioms}

In this section, we show how the graphical jumps presented above can be encoded in cut-free \(\MLL\) proof structures using 
the axiom link of fresh atomic formulas $\J$ to create synchronisation points.
We first introduce the idea for a single jump, then generalize it to multiple jumps.

\subsection{Case of a single jump}
\begin{definition}
  \label{def:gadget1}
  Let \(F\) and \(F'\) be two formul\ae{}. A \emph{jump} from \(F\) to \(F'\) is
  the gadget:
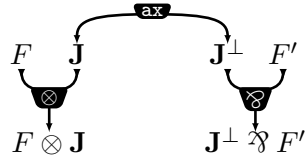
\vspace{\beforepn}
  \begin{proofnet}
    \pnformulae{
      \pnf[F]{$F$}~\pnf[j]{$\J$}~~~~\pnf[jb]{$\J^{\bot}$}~\pnf[F']{$F'$}\\
    }
    \pnaxiom{j,jb}
    \pntensor{F,j}{$F \otimes \J$}
    \pnpar{F',jb}{$\J^{\bot}\parr F'$}
  \end{proofnet}
  
  Let \(R\) be a proof-structure, and \(F\), \(F'\) be two formula occurrences 
  in \(R\).

  We define \(\J_{F\to F'}(R)\), the \emph{proof structure with a jump} from \(F\) to \(F'\), by
  grafting in $R$ the jump from \(F\) to \(F'\) where \(F\) and \(F'\) are, and making
  all the types coherent by propagating downward the formula labels.
\end{definition}

Let us illustrate the previous definition. Consider the (correct) proof-structure \(R\) in Example~\ref{ex:ps-two-seq}, we
can add a jump from \(F = (A \parr B)\otimes C^{\bot}\) to \(C\), resulting in
the proof-structure \(\J_{F\to C}(R)\), where the conclusion formulas have been updated to take into account the impact of adding the jump: 
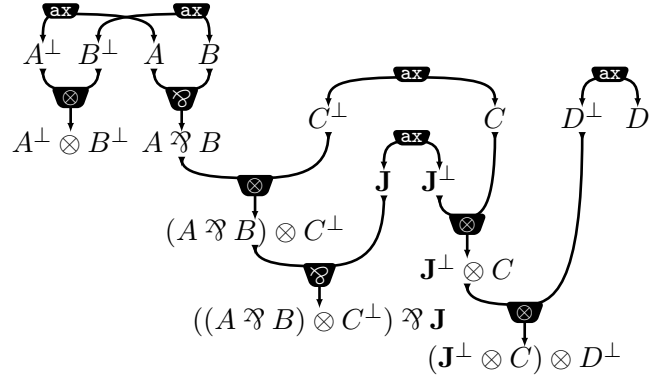
\vspace{\beforepn}
  \begin{proofnet}
  \pnformulae{
    \pnf[Ab]{$A^{\bot}$}~\pnf[Bb]{$B^{\bot}$}~\pnf[A]{$A$}~\pnf[B]{$B$}\\
    ~~~~~~\pnf[Cb]{$C^{\bot}$}~~~\pnf[C]{$C$}~~\pnf[Db]{$D^{\bot}$}~\pnf[D]{$D$}\\
    ~~~~~~~\pnf[J]{$\J$}~\pnf[Jb]{$\J^{\bot}$}
  }
  \pnaxiom{Ab,A}[1][-0.5]
  \pnaxiom{Bb,B}[1][0.5]
  \pnaxiom{Cb,C}
  \pnaxiom{Db,D}
  \pnaxiom{Jb,J}
  \pntensor{Ab,Bb}[AboBb]{$A^{\bot}\otimes B^{\bot}$}
  \pnpar{A,B}[ApB]{$A \parr B$}
  \pntensor{ApB,Cb}[ApBoCb]{$(A \parr B)\otimes C^{\bot}$}
  \pntensor{Jb,C}[JboC]{$\J^{\bot} \otimes C$}
  \pntensor{JboC,Db}{$(\J^\bot \otimes C) \otimes D^{\bot}$}
  \pnpar{ApBoCb,J}{$((A \parr B)\otimes C^{\bot})\parr \J$}
  \end{proofnet}

This proof-structure is much more rigid than \(R\). It also has only two
sequentializations:
\[
\scalebox{1}{$
  \begin{prooftree}
    \hypo{\vdash A^{\bot},A}
    \hypo{\vdash B^{\bot},B}
    \infer2[\(\otimes\)]{\vdash A^{\bot}\otimes B^{\bot},A,B}
    \infer1[\(\parr\)]{\vdash A^{\bot}\otimes B^{\bot},A\parr B}
    \hypo{\vdash C^{\bot},C}
    \infer2[\(\otimes\)]{\vdash A^{\bot}\otimes B^{\bot}, (A \parr B)\otimes C^{\bot}, C}
    \hypo{\vdash \J^{\bot},\J}
    \infer2[\(\otimes\)]{\vdash A^{\bot}\otimes B^{\bot}, (A \parr B)\otimes
      C^{\bot},\J, \J^{\bot} \otimes C}
    \infer1[\textcolor{green}{\(\parr\)}]{\vdash A^{\bot}\otimes B^{\bot}, (A \parr B)\otimes
      C^{\bot}\parr \J, \J^{\bot} \otimes C}
    \hypo{\vdash D^{\bot},D}
    \infer2[\(\otimes\)]{\vdash A^{\bot}\otimes B^{\bot}, (A \parr B)\otimes
      C^{\bot}\parr \J, (\J^{\bot} \otimes C)\otimes D^{\bot},D}
  \end{prooftree}
  $}
\]
and the one where the \(\parr\) rule shown in green is commuted below the
\(\otimes\) rule. They differ in a much lighter way than the two
sequentializations of \(R\): indeed, forgetting every occurrence of \(\J\) in
both of these two sequent calculus proofs yields only one of the two
sequentializations of the original proof-net. The presence of the jump from
\(F = (A \parr B)\otimes C^{\bot}\) to \(C\) forces \(F\) to be above \(C\) in
every sequentialization of \(\J_{F\to C}(R)\). Let us turn this example into a
general theorem.

\paragraph{$\MLLJ$}
\label{para:mllj}

From now on, we consider $\J$ as a special atom that does not appear in regular
$\MLL$ formulae and proofs.
We write $\MLLJ$ (respectively $\PSJ$) for the cut-free $\MLL$ proof system (resp.
proof-structures) that does allow for the presence of \(\J\) but with strong
restrictions on its use described below.

We will call \(J\)-formul\ae{} (respectively \(J\)-environments), denoted by
$J, K, \ldots$ (respectively $\Gamma_\J,\Delta_\J,\Theta_\J, \cdots$), formul\ae{} in the
grammar $J = \J\mid \J^\perp$ (respectively environments consisting of
\(J\)-formul\ae{}). $\MLLJ$ has the same rules as $\MLL$ except for
$J$-formulae that are kept separated from other formulae until being used:
\[
  \begin{prooftree}
    \hypo{\vdash \Gamma, F, \Theta_\J}
    \infer0[ax$_\J$]{\vdash \J, \J^\perp}
    \infer2[\(\otimes_\J, F\notin J\)]
    {\vdash \Gamma, F\otimes \J, \J^\perp, \Theta_\J}
  \end{prooftree}
  \qquad  \qquad
  \begin{prooftree}
    \hypo{\vdash \Gamma, F, \J^\perp,\Theta_\J}
    \infer1[\(\parr_{\J^\perp}, F\notin J\)]{\vdash \Gamma, F\parr  \J^\perp,\Theta_\J}
  \end{prooftree}
\]

Proof structures and proof nets for $\MLLJ$ are essentially the same as in 
  $\MLL$ except for the typing restriction put on links $\otimes$ and $\parr$ 
  when dealing with $J$-formul\ae{}.
  This motivates to define   $\MLL_\J$-formulas by the following grammar:
%
  $$F_J = A \mid F_J \parr F_J\mid F_J \otimes F_J \mid F_J\parr \J^\perp \mid F_J\otimes \J.$$
The following proposition is trivial by induction on the structure of $\MLL_\J$ proofs.
\begin{proposition}
  If $\vdash \Gamma$ is derivable in $\MLL_\J$, then $\Gamma$ can be partitioned in $\Delta,\Theta_\J$  such that every formula in $\Delta$ is a $\MLL_J$-formula
  and every
formula in $\Theta_\J$ is $\J^\perp$.
\end{proposition}

In particular, when removing $\J$, $\MLLJ$ is equivalent to $\MLL$. More
precisely:
\ctodo{Alexis: What is an $\MLL_\J$ formula? I think it is the $F_\J$ grammar I defined in the proposition, but would like to double-check... }
\ctodo{Aurore: $\MLLJ$ c'est à la fois les $F_J$ et les $J$ formulae}
\ctodo{Alexis: mais si tu prends aussi les J, la définition d'effacement ne marche pas, je crois qu'il faut vraiment prendre uniquement les $F_\J$.}
\begin{definition}
  \label{def:erasure}
  Let \(F\) be a $\MLLJ$-formula, the \emph{erasure} \(\erasure{F}{\J}\) of
  \(\J\) in \(F\) is inductively defined by:
  \begin{itemize}
  \item if \(F\) is an atom:  \(\erasure{F}{\J} = F\);
  \item if $F = F_0 \boxempty F_1$ with $\boxempty \in \{\parr, \otimes\}$:
    $\erasure{F_0}{\J}$ if $F_{1} \in \{\J, \J^\perp\}$, $\erasure{F_0}{\J} \boxempty \erasure{F_1}{\J}$ otherwise;
  \end{itemize}
  Let \(\pi\) be an \(\eta\)-expanded sequent calculus proof in $\MLLJ$ (that is, all axiom rules are applied to atomic formulas only), the
  \emph{erasure} \(\erase{\pi}{\J}\) is inductively defined by:
  \begin{itemize}
  \item if \(\pi\) is an axiom \(
    \begin{prooftree}
      \hypo{\vdash B^{\bot},B}
    \end{prooftree}
    \)
    then \(\erase{\pi}{\J} = \pi\);
  \item if \(\pi\) is an application of a \(\parr\) on $A\parr B$, with subproof $\pi'$:
    \begin{itemize}
      \item if \(B\neq\J^\perp\), \({\erase{\pi}{\J}}\) is defined by applying a $\parr$ rule on $\erasure{A}{\J}$ and $\erasure{B}{\J}$ to \({\erase{\pi'}{\J}}\);
    \item if \(\pi\) is an application of a \(\parr_{\J^\perp}\) after $\pi'$
    then \(\erase{\pi}{\J}\) is defined as: \({\erase{\pi'}{\J}}\);
    \end{itemize}
  \item if \(\pi\) is an application of a \(\otimes\) on $A\otimes B$, combining subproofs $\pi_A$ and $\pi_B$:
    \begin{itemize}
    \item if \(B \neq \J\), then  \({\erase{\pi}{\J}}\) is defined by combining
       \({\erase{\pi_A}{\J}}\) and  \({\erase{\pi_B}{\J}}\) with a $\otimes$ rule on $\erasure A \J$ and $\erasure B \J$;
    \item if \(\pi\) is an application of a \(\otimes_\J\) combining 
    $\pi'$ and $\ax_\J$ 
    then \(\erase{\pi}{\J}\) is defined as \(\erase{\pi'}{\J}\).
\end{itemize}
  \end{itemize}
    Let $R$ be an $\MLLJ$-proof-structure the \emph{erasure}, \(\erase{R}{\J}\), is inductively defined in the same way 
  as \(\erase{\pi}{\J}\), that is by erasing every occurrence of $\J,\J^\perp$ and the $\parr$ and $\otimes$ link immediately connected to those atoms as well as the axioms links between $\J$ and $\J^\perp$.
\end{definition}

  Note that if $\vdash \Gamma$ is the conclusion of $\pi$ (resp. $R$), then $\vdash \erasure{\Gamma}{\J}$
  is the conclusion of $\erase{\pi}{\J}$ (resp. $\erasej{R}$). Erasures cannot be 
  empty by construction of $\MLL_J$.
  Moreover, if $\pi \in \seq(R)$ then $\erasej{\pi} \in \seq(\erasej{R})$.
%
%
%

  $\MLL$ is included in $\MLLJ$.
  Moreover:
  
 \begin{proposition}
If $\vdash \Gamma,\Theta_J$ is provable in $\MLLJ$ 
   with $\vdash \Gamma \neq\emptyset$, 
   then $\erasurej{\Gamma}$ is provable in $\MLL$.
   In fact, for $\pi \vdash \Gamma,\Theta_\J$, $\erasej{\pi}$
   is a $\MLL$ proof and there is an embedding 
   $\erasej{\pi} \hookrightarrow \pi$ 
   of the rules of $\erasej{\pi}$ onto the rules of $\pi$, preserving their kind and
   the order between them.
 \end{proposition}
 \begin{proof}
   By structural induction on $\pi$.
   Every case is straightforward, except maybe when $\otimes$ is the
   last rule of $\pi$: 
   if $F$ and $G$ are the active formulae, one shall remark 
   that by construction $\erasurej{F}$ or $\erasurej{G} \neq \emptyset$.  
 \end{proof}

\begin{theorem}
\label{thm:jump-as-axiom}
  Let \(R\) be a proof-structure, and \(F\) and \(F'\) two formula occurrences, then
  \(\erasej{\seq\ (\J_{F\to F'}(R))} =  \seq_{(F, F')}\ (R)\).
  In other words, the following diagram commutes:
  \begin{center}
    \begin{tikzcd}
      \PSJ \ar[r,"\seq"] & \mathcal{P}(\MLLJ) \ar[d,"\mathcal{P}(\erasej{\_})"]\\
      \PS  \ar[u,"\J_{F\to F'}(\_)"] \ar[r,"\seq_{(F,F')}"]& \mathcal{P}(\MLL)
    \end{tikzcd}
  \end{center}
\end{theorem}

In particular, correctness of the proof-net enriched with a jump is euivalent to existence of particular sequentializations:

\begin{corollary}[correctness]
  Let \(R\) be a proof-structure, and \(F\) and \(F'\) two formula occurrences, then
  \(\J_{F \to F'}(R)\) is correct if and only if there exists
  a sequentialization of \(R\) such that \(F\) is above \(F'\).
\end{corollary}

Let us prove Theorem~\ref{thm:jump-as-axiom}.
\begin{proof}
  $\subseteq$.
  From Def.~\ref{def:erasure} it is immediate that 
    $\erasej{\seq(R)} \subseteq \seq(\erasej{R})$.
    Now, let \(\pi\) be any sequentialization of \(R_{F\to F'}\).
    \(\pi\) contains an axiom \(\vdash \J^{\bot}, \J\), to which is
  applied a \(\otimes\)-rule tensoring \(J\) with \(F\).
  So $F$ has been created in the second branch leaving this $\otimes$-rule.
  At one point below, there is a \(\parr\)-rule between \(F'\) and \(\J^{\bot}\).
  So the use of $F'\parr \J^\bot$ is below the creation of $F$ in $\pi$, 
  hence so is the use of $F'$ in $\erasure{\pi}{J}$.

  $\supseteq$.
  Let $\pi \in \seq_{(F,F')} (R)$. By hypothesis, the rule $\mathsf{r}$ introducing $F$ is above the rule $\mathsf{r'}$ 
using $F'$ in $\pi$, we can thus perform the following transformation 
adding three more rules in $\pi$ (and propagating the changes in the resulting sequent accordingly):

\[ 
    \pi = 
      \begin{prooftree}
        \hypo{\vdots}
        \infer1[$\mathsf{r}$]{\vdash \Gamma, F}
        \infer1{\vdots}
        \infer1{\vdash \Gamma', F'}
        \infer1[$\mathsf{r'}$]{\vdots}
      \end{prooftree}
    \qquad \to \qquad
      \begin{prooftree}
        \hypo{\vdots}
        \infer1[$\mathsf{r}$]{\vdash \Gamma, F}
        \infer0[$\mathsf{ax}$]{\vdash \J, \J^{\bot}}
        \infer2[\(\otimes\)]{\vdash F \otimes \J, \Gamma,\J^{\bot}}
        \infer1{\vdots}
        \infer1{\vdash \Gamma'[F\otimes\J/F], F',\J^\perp}
        \infer1[\(\parr\)]{\vdash \Gamma'[F\otimes\J/F], F'\parr\J^\perp  }
        \infer1[$\mathsf{r'}$]{\vdots}
      \end{prooftree}
      \]
      
    These 3 rules directly map to the jump gadget: the resulting proof is a sequentialization of $\J_{F\to F'}(R)$.
\end{proof}

We can now model a single jump in a proof net. This situation
is far too restricted: we now move to the modelling of multiple jumps.

\subsection{Multiple jumps}
The definition of jump given above suffers from an obvious drawback: there is no
canonical way to consider two jumps having the same source or destination. 
To solve this issue, we need the ability to 
\begin{enumerate}
\item parallelize: jumps with the same source or destination have no order,
  hence should not be represented as successive applications of \(\parr\) or
  \(\otimes\) rules;
\item contract: each jump needs two contractions: one on the source, one on the
  target.
\end{enumerate}
Following these two requirements we ask $\J$ to have the following contraction properties:
\[
  \underbrace{\J\otimes \ldots \otimes \J}_{n\geq 0} \multimap \J
\qquad
  \underbrace{\J^\perp \parr \ldots \parr \J^\perp}_{m > 0} \multimap \J^\perp
\]
Hence, we disallow axioms on \(J\)-formul\ae{} and extend \(\MLLJ\) with the
rules, for \(n\geq 0, m>0\):
\[
  \begin{prooftree}
    \infer0[$\J(m)$]{\vdash \J, \underbrace{\J^\perp, \ldots, \J^\perp}_m}
  \end{prooftree}
  \qquad \quad
  \begin{prooftree}
    \hypo{\vdash \Gamma, \overbrace{\J^\perp, \ldots, \J^\perp}^n}
    \infer1[$\J(n)$]{\vdash \Gamma, \J^\perp}
  \end{prooftree}
\]
and, in the same way, we add the following cells to \(\PSJ\):
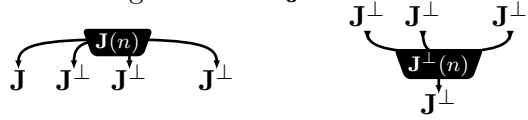
\vspace{\beforepn}
\begin{proofnet}
  \pnformulae{
    ~~~~~~~~\pnf[jb1']{$\J^{\bot}$}~\pnf[jb2']{$\J^{\bot}$}~~\pnf[jbn']{$\J^\bot$}\\
    \pnf[j]{$\J$}~\pnf[jb1]{$\J^{\bot}$}~\pnf[jb2]{$\J^{\bot}$}~~\pnf[jbn]{$\J^\bot$}\\
  }
  \pnjaxiom[ax]{n}{j,jb1,jb2,jbn}
  \pnjcontr{n}{jb1',jb2',jbn'}{$\J^{\bot}$}
\end{proofnet}

In the $\otimes_\J$ rule, we also replace Ax$_\J$ by $\J(n)$ for any $n\geq0$.

Note that having a weakening on $\J$ ({\it i.e.} the rule $\J(0)$) is not required when 
dealing with cut-free proofs but we simply anticipate on the next section. 

We generalise the encoding of a single jump to a list of jumps by making space 
to connect any edge with another one and then effectively realising the connections
described in the list.

\begin{definition}
  \label{def:multiple}
Given a proof-structure $R$, we define the \emph{proof structure with synchronisation  
points} $\J(R)$ as the proof structure $R$ in which every edge $F$ of $R$ 
is replaced by the following gadget $\J(F)$
(relabelling its successors to make all the 
types coherent):
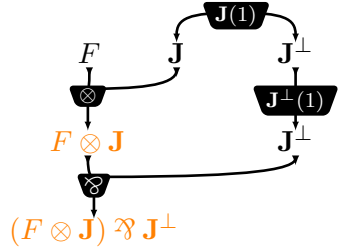
\vspace{\beforepn}
\begin{proofnet}
  \pnformulae{
    \pnf[f]{$F$}~~\pnf[j]{$\J$}~~~\pnf[jb]{$\J^\bot$}\\
  }
  \pnjaxiom[ax]{1}{j,jb}
  \pntensor[tensor]{f,j}{$\textcolor{orange}{F \otimes \J}$}[1][-0.6cm]
  \pnjcontr{1}{jb}[cjb]{$\J^\bot$}
  \pnpar{tensor,cjb}{$\textcolor{orange}{(F \otimes \J)\parr \J^{\bot}}$}[1][-1.3cm]
\end{proofnet}


The $\otimes$-link bounded to the \(\J(1)\) axiom is the
\emph{out-synchronisation point} of \(F\), while the $\parr$-link bounded
to the \(\J^{\bot}(1)\) contraction is its
\emph{in-synchronisation point}.

Given two formula occurrences $F$ and $F'$ in $R$, a \emph{jump from $F$ to $F'$} is
then encoded in $\J(R)$ by connecting the out-synchronisation point of $F$ to the
in-synchronisation point of $F'$, that is, by adding a $\J^\perp$ edge from
the $\J$ axiom link of $\J(F)$ to the $\J^\perp$ contraction link of $\J(F')$.

Finally, given $\L$ a set of pairs of formula occurrences in $R$, 
we define $\J_{\L}(R)$ the proof structure corresponding to $\J(R)$
with jumps between each pair of formula occurrences provided by $\L$.
\end{definition}

\begin{figure}[b] 
  \caption{Multiple jumps}
  \label{fig:multiple}
  \vspace{\beforepn} \vspace{4pt}
  \begin{proofnet}
    \pnformulae{
      \pnf[f]{$F$}~~\pnf[j]{$\J$}~\pnf[jb]{$\J^\bot$}~\pnf{$\cdots$}~~
      \pnf[jbmid]{$\J^\bot$}~~\pnf{$\cdots$}~
      \pnf[jb']{$\J^\bot$}~\pnf[j']{$\J$}~~\pnf[f']{$F'$}\\
    }
    \pnjaxiom[ax]{n+1}{j,jb,jbmid}
    \pntensor[tensor]{f,j}{$\textcolor{orange}{F \otimes \J}$}[1][-0.6cm]
    \pnjcontr{1}{jb}[cjb]{$\J^\bot$}
    \pnpar{tensor,cjb}{$\textcolor{orange}{(F \otimes \J)\parr \J^{\bot}}$}[1][-1.1cm]
    \pnjaxiom[ax]{1}{j',jb'}
    \pntensor{f',j'}[tensor']{$\textcolor{orange}{F' \otimes \J}$}[1][0.6cm]
    \pnjcontr{m+1}{jbmid,jb'}[cjb']{$\J^\bot$}
    \pnpar{tensor',cjb'}{$\textcolor{orange}{(F' \otimes \J)\parr \J^{\bot}}$}[1][1.5cm]
  \end{proofnet}
\end{figure}

Figure~\ref{fig:multiple} illustrates the encoding of a jump from $F$ 
to $F'$ where $F$ also has $n$ outgoing jumps and $F'$ also has $m$ 
incoming jumps. In the sequel we will sometimes hide these jumps-as-axioms
encoding and schematize them with dashed arrow, which in the case
of Figure~\ref{fig:multiple} would correspond to:
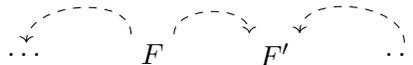
\begin{center}
  \begin{tikzpicture}
    \node (f) {$F$};
    \node[right=of f] (f') {$F'$};
    \node[right=of f'] (others2) {\textellipsis};
    \node[left=of f] (others) {\textellipsis};
    \draw[->,dashed] (f.north east) to[out=90,in=90] (f'.north west);
    \draw[->,dashed] (others2.north) to[out=90,in=90] (f'.north east);
    \draw[->,dashed] (f.north west) to[out=90,in=90] (others.north);
  \end{tikzpicture}
\end{center}
\ctodo{Aurore: il y aurait moyen de rajouter les traits correspondant
au edge de F et F' dans le réseau ?}

\begin{remark}
  It is striking to note how close our encoding of jumps can be from 
  double negation translations used to add sequentializations in the translation
  of classical logic into intuitionistic logic. Indeed, in linear logic, 
 double negation translation correspond to $(A\multimap \bot) \multimap \bot$
  that is $(A\otimes \One) \parr \bot$. This can be understood as
  the production of a token ($\One$) once $A$ is created, and the wait
  for this token to return before $A$ can be consumed. In the above encoding
  of jumps, the $\J$ atom replaced $\One$, which allows many tokens to be 
  produced and given to the environment after the creation of $A$ 
  (these tokens furthermore carry the information of coming from $A$
  thanks to the axiom link), while on the other side $A$ can now wait 
  for many specific tokens (coming from other formulae) to be put in the environment
  before it can actually be used.
\end{remark}
\ctodo{aurore : je ne sais pas si cette remarque est très claire..
  Luc : Est-ce que ça dit des choses intéressantes sur la goi des réseaux avec preuves ?}

The definition of $\erasej(\_)$ (see~\ref{def:erasure}) does not change and
following the same reasoning as for Theorem~\ref{thm:jump-as-axiom} we have:
\begin{theorem}
  \label{thm:jumps-as-axiom}
  Let \(R\) be a cut-free proof-structure, and $\L$ be a set
  of pairs of formula occurrences,  then
  \(\erasej{\seq\ (\J_{\L}(R))} =  \seq_{\L}\ (R)\).
  In other words,
  \ctodo{Aurore commencer par le diagramme et dire that is}
  \begin{center}
    \begin{tikzcd}
      \PSJ \ar[r,"\seq"] & \mathcal{P}(\MLLJ) \ar[d,"\mathcal{P}(\erasej{\_})"]\\
      \PS  \ar[u,"\J_{\L}(\_)"] \ar[r,"\seq_{\L}"]& \mathcal{P}(\MLL)
    \end{tikzcd}
  \end{center}
\end{theorem}

\begin{corollary}[correctness]
  Let \(R\) be a proof-structure, and $\L$ be a set of pairs of formula
  occurrences,  then
  \(\J_{\L}(R)\) is correct if and only if there exists
  a sequentialization of \(R\) such that \(F\) is above \(F'\) for every $(F,F')\in\L$.
\end{corollary}

\begin{remark}
  In~\cite{DiGiamberardinoFaggian:journal}, jumps are defined as edges from one link to one of the \emph{entry port} of another link: a parr link having one entry port
  to which each of its premises are plugged;
  a tensor link having two entry ports, one for each of its premises. 
  The correctness criterion for $\MLL$ proof structures with such graphical jumps is then 
  a recast of the usual Danos-Regnier criterion: a proof structure is correct if
  and only if every of its \emph{switching paths} is connected and acyclic, a switching path
  being here defined as a path which does not use any two edges
  entering trough the same port of a $\parr$ link\ctodo{Alex: il n'y a pas un parr manquant ici?}.

  In the Danos-Regnier criterion for $\MLL$ proof structures, switching paths 
  are defined as paths which never follow both premises of a $\parr$ link.
  Viewing our in-synchronisation point as their entry port, the usual Danos-Regnier
  criterion on our encoding of jumps correspond \emph{exactly} to their recast of the 
  same criterion.
\end{remark}

\begin{remark}
  In~\cite{hughes18-unets} Hughes introduces a $\MLL$-encoding of his 
  unification nets in order to prove their correction. Although we were not 
  aware of it at the time we wrote this paper it must be acknowledged that this
  encoding is morally the same as our gadget from Definition~\ref{def:gadget1}
  Being only interested in an encoding that preserves splitting tensors and
  switching cycles (\textit{i.e.} correction), Hughes does not develop on the
  sequentialisation aspect of his encoding: in particular he shows no connection
  between his encoding and the set of sequentialisations of its nets, multiple
  cuts are not treated canonically and cuts are treated as tensors, evacuating
  any questions about their dynamics.
\end{remark}


\section{Cuts}
\label{sec:cuts}

We now investigate the internalisation of jumps in the presence of cuts. In
particular, we need to define the dynamics of $\J$ and \(\J^{\bot}\).

As intuited by~\cite{DiGiamberardinoFaggian:conf}, we expect jumps to propagate
through cut elimination: suppose indeed that the left-hand side of a cut (on a
formula \(F\)) has jumps coming into it, while the right-hand side
(\(F^{\bot}\)) has jumps coming out of it. It forces, \emph{in any
  sequentialization}, the sub-formul\ae{} of the cut-formula to be above what is
below the cut. 
This information has to be kept in some way during cut-elimination, even when
the considered cut vanishes. Moreover, if \(n\) jumps enter into \(F\) and \(m\)
jumps go out of \(F^{\bot}\), the reduction of the cut will produce
\(m\times n\) jumps: one for each composition of jumps. This reveals the
exponential nature of jumps: each jump arriving onto one of the two cut formulae
must be propagated towards every formula occurrences reached by a jump leaving
one of the two cut formulae, which cannot be realized by the multiplicative
combinatorics. Unsurprisingly we define a commutation rule between
\(\J^{\bot}(m)\) and \(\J(n)\) as: \ctodo{Aurore :commutation ou reduction ?}
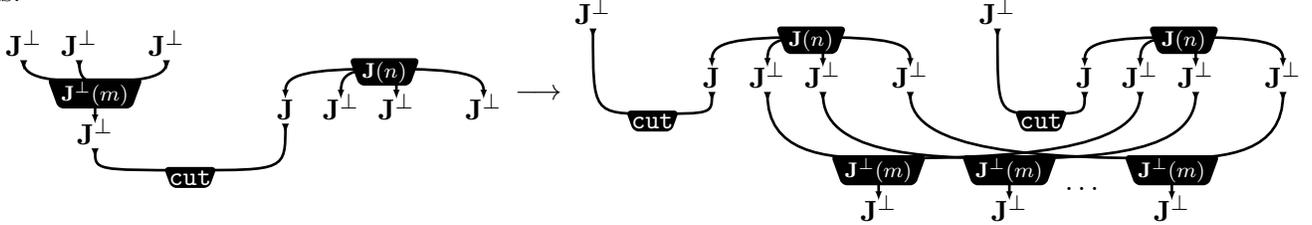
\vspace{\beforepn}
\begin{center}
  \label{eq:ctr}
  \pnet{
    \pnformulae{
      \pnf[jb1']{$\J^{\bot}$}~\pnf[jb2']{$\J^{\bot}$}~~\pnf[jbn']{$\J^\bot$}\\
      ~~~~~~\pnf[j]{$\J$}~\pnf[jb1]{$\J^{\bot}$}~\pnf[jb2]{$\J^{\bot}$}~~\pnf[jbn]{$\J^\bot$}\\
    }
    \pnjaxiom[ax]{n}{j,jb1,jb2,jbn}
    \pnjcontr{m}{jb1',jb2',jbn'}[contr]{$\J^{\bot}$}
    \pncut{j,contr}
  }
  \(\longrightarrow\)
  \pnet{
    \pnformulae{
      \pnf[jb1']{$\J^{\bot}$}~~~~~~~~~\pnf[jb2']{$\J^{\bot}$}\\
      ~~~\pnf[j_1]{$\J$}~\pnf[jb1_1]{$\J^{\bot}$}~\pnf[jb2_1]{$\J^{\bot}$}~~\pnf[jbn_1]{$\J^\bot$}~~~~\pnf[j_2]{$\J$}~\pnf[jb1_2]{$\J^{\bot}$}~\pnf[jb2_2]{$\J^{\bot}$}~~\pnf[jbn_2]{$\J^\bot$}\\
      \\
~~~~~~~~~~~    \pnf{$\dots$}
    }
    \pnjaxiom[ax]{n}{j_1,jb1_1,jb2_1,jbn_1}
    \pnjaxiom[ax2]{n}{j_2,jb1_2,jb2_2,jbn_2}
    \pncut{j_1,jb1'}
    \pncut{j_2,jb2'}
    \pnjcontr{m}{jb1_1,jb1_2}[contr]{$\J^{\bot}$}[1.6][-1cm]
    \pnjcontr{m}{jb2_1,jb2_2}[contr]{$\J^{\bot}$}[1.6]
    \pnjcontr{m}{jbn_1,jbn_2}[contr]{$\J^{\bot}$}[1.6][1cm]
  }
\end{center}

which corresponds to the following map from the monoidal structure of any
monoidal category interpreting \(\MELL\) proof-structures ($\gamma$ stands for
the commutation):
\[\ctr(m)^\perp \circ \gamma \circ \ax(n) = 
  \bigotimes_n \J 
  {\multimap}
  \bigotimes_n \bigotimes_m \J 
  {\multimap}
  \bigotimes_m \bigotimes_n \J 
  {\multimap}
  \bigotimes_m  \J 
\]

To propagate newly created jumps from the above reduction rule, we also need
$\J(n)$ and $\J^{\bot}(m)$ to both act as a monoid: the corresponding rules are  given in Figure~\ref{fig:propag-jumps} (\(m\geq 0, n >0\)).
\begin{figure}
  \begin{equation}
    \label{red:axn}
\vspace{-1cm}
  \pnet{
    \pnformulae{
      \pnf[j_1]{$\J$}~\pnf[jb1_1]{$\J^{\bot}$}~\pnf[jb2_1]{$\J^{\bot}$}~~\pnf[jbn_1]{$\J^\bot$}~~~~\pnf[j_2]{$\J$}~\pnf[jb1_2]{$\J^{\bot}$}~\pnf[jb2_2]{$\J^{\bot}$}~~\pnf[jbn_2]{$\J^\bot$}\\
    }
    \pnjaxiom[ax]{n}{j_1,jb1_1,jb2_1,jbn_1}
    \pnjaxiom[ax2]{m}{j_2,jb1_2,jb2_2,jbn_2}
    \pncut{jbn_1,j_2}
  }
\quad  \longrightarrow\quad
  \pnet{
    \pnformulae{
      \pnf[j_1]{$\J$}~\pnf[jb1_1]{$\J^{\bot}$}~\pnf[jb2_1]{$\J^{\bot}$}~\pnf[jbn_1]{$\J^\bot$}~\pnf[jb1_2]{$\J^{\bot}$}~\pnf[jb2_2]{$\J^{\bot}$}~~\pnf[jbn_2]{$\J^\bot$}\\
    }
    \pnjaxiom[ax]{n+m-1}{j_1,jb1_1,jb2_1,jbn_1,jb1_2,jb2_2,jbn_2}
  }
\end{equation}
\begin{equation}
   \label{eq:ctr}
  \pnet{
    \pnformulae{
      \pnf[jb1']{$\J^{\bot}$}~\pnf[jb2']{$\J^{\bot}$}~~\pnf[jbn']{$\J^\bot$}\\
      ~~~~\pnf[jb1]{$\J^{\bot}$}~\pnf[jb2]{$\J^{\bot}$}~~\pnf[jbn]{$\J^\bot$}\\
    }
    \pnjcontr{m}{jb1',jb2',jbn'}[cjb']{$\J^{\bot}$}
    \pnjcontr{n}{cjb',jb1,jb2,jbn}[cjb]{$\J^{\bot}$}
  }
\quad  \longrightarrow\quad
  \pnet{
    \pnformulae{
      \pnf[jb1']{$\J^{\bot}$}~\pnf[jb2']{$\J^{\bot}$}~~\pnf[jbn']{$\J^\bot$}
      ~~~~\pnf[jb1]{$\J^{\bot}$}~\pnf[jb2]{$\J^{\bot}$}~~\pnf[jbn]{$\J^\bot$}\\
    }
    \pnjcontr{n+m-1}{jb1',jb2',jbn',jb1,jb2,jbn}[cjb']{$\J^{\bot}$}
  }
\end{equation}
\caption{Reductions rules for propagating jumps\label{fig:propag-jumps}}
\end{figure}

%
\paragraph{Bisimulation}
On top of propagating the jumps arriving on the premises of two links involved in
a cut, to the target of the jumps leaving the conclusions of these exact same links,
we also need to propagate jumps with respect to the sub-formula ordering, and to
destroy the ones that arrive onto the premises of a cut being eliminated.

With this intuition in mind, we extend the definition of $\J_\L(\_)$ 
presented in Section~\ref{sec:jumps-as-axioms} in order to 
preserve the duality of (cut) formulae and ensure the 
propagation of jumps throughout cut elimination. We add propagation 
and absorption points to the previous synchronisation points: 

\begin{definition}
  Given a proof-structure $R$, we define the \emph{proof structure with
    synchronisation points} $\J(R)$ as the proof structure $R$ in which the
    outgoing edges of axiom, tensor (\(\otimes\)) and par (\(\parr\)) links are
  replaced (from top to bottom) by the gadgets depicted in
  Figures~\ref{fig:gadget-axiom}, \ref{fig:gadget-pos} and \ref{fig:gadget-neg}
  respectively (making all the types coherent, \(\J(A)\) is the type of gadget
  interpreting \(A\)).
  \begin{figure}[t]
    \centering
    \vspace{\beforepn}
    \begin{proofnet}
      \pnformulae{
        \pnf[a]{$A$}~~
        \pnf[j]{$\J$}
        ~~\pnf[jb2]{$\J^{\bot}$}
        ~~\pnf[jb4]{$\J^{\bot}$}  ~~\pnf[jb3]{$\J^{\bot}$} 
        ~~\pnf[ab]{$A^{\bot}$}\\
        \\
        \\
        \\
        ~~~~~~~~~~~~~~\pnf[jz]{$\J$}\\
        \\
        \pnf[jz']{$\J$}
      }
      \pnjaxiom{3}{j,jb2,jb4,jb3}
      \pnaxiom{a,ab}[1.5]
      \pntensor{j,a}[joa]{$\textcolor{orange}{\J\otimes A}$}[2.2]
      \pnjcontr{1}{jb2}[cjb2]{$\J^{\bot}$}
      \pnjcontr{1}{jb3}[cjb3]{$\J^{\bot}$}
      \pnjcontr{1}{jb4}[cjb4]{$\J^{\bot}$}
      \pnpar{cjb3,ab}[jbpab]
      {$\textcolor{green}{\J^{\bot}\parr A^{\bot}}$}
      \pnpar{cjb2,joa}[jbpjbpjoa]
      {$\textcolor{orange}{(\J\otimes A)\parr \J^{\bot}}$}
      \pnjaxiom{0}{jz}
      \pntensor{jbpab,jz}[jbpabojoj]
      {$\textcolor{red}{(\J^{\bot}\parr A^{\bot})\otimes \J}$}
      \pnpar{jbpabojoj,cjb4}
      {$\textcolor{orange}{((\J^{\bot}\parr A^{\bot})\otimes \J)\parr \J^{\bot}}$}
      \pnjaxiom{0}{jz'}
      \pntensor{jbpjbpjoa,jz'}
      {$\textcolor{red}{((\J\otimes A)\parr \J^{\bot})\otimes\J}$}
    \end{proofnet}
 
    \caption{Gadget for the axiom $A$ and $A^{\bot}$}
    \label{fig:gadget-axiom}
  \end{figure}
  \begin{figure}[t]
    \centering
    \vspace{\beforepn}
    \begin{proofnet}
      \pnformulae{
        ~~~~~~~~~~\pnf[a]{$A$}~\pnf[b]{$B$}\\
        ~~~~~~~~~\pnf[j1]{$\J$}~\pnf[ja]{$\vdots$}~\pnf[jb]{$\vdots$}~\pnf[j2]{$\J$}\\
        \\
        \\
        \\
        \pnf[inc]{\textcolor{blue}{$\bullet$}}~\pnf[jb1]{$\J^{\bot}$}~\pnf[jb2]{$\J^{\bot}$}~\pnf[jb3]{$\J^{\bot}$}~\pnf[j3]{$\J$}\\
        \\
        \\
        ~~~~~~~~~~~~~\pnf[j4]{$\J$}~~~\pnf[jb4]{$\J^{\bot}$}~~\pnf[inc2]{$\textcolor{orange}{\bullet}$}\\
        ~~~~~~~~~~~~~~~~~~~~~\pnf[inc3]{$\textcolor{orange}{\bullet}$}\\
        \\
        ~~\pnf[j5]{$\J$}
      }
      \draw[-,dotted] (a) to (ja);
      \draw[-,dotted] (b) to (jb);
      \pntensor{j1,ja}[joa]{$\vdots$}
      \pntensor{j2,jb}[job]{$\vdots$}
      \pntensor{joa,job}[joaojob]{$\J(A)\otimes \J(B)$}
      \pntensor{joaojob,j3}[joaojoboj]
      {$\textcolor{green}{\J(A)\otimes \J(B)\otimes \J}$}
      \pnjaxiom{1}{jb3,j3}
      \pnjaxiom{1}{jb2,j1}
      \pnjaxiom{1}{jb1,j2}[2]
      \pnjcontr{3}{inc,jb1,jb2,jb3}[cjb]{$\J^{\bot}$}
      \pnpar{cjb,joaojoboj}[jbpjoaojoboj]
      {$\textcolor{blue}{\J^{\bot}\parr (\J(A)\otimes \J(B) \otimes \J)}$}
      \pnjaxiom{1}{jb4,j4,inc2}
      \pntensor{jbpjoaojoboj,j4}[jbpjoaojobojoj]
      {$\textcolor{orange}{(\J^{\bot}\parr (\J(A)\otimes \J(B)\otimes \J))\otimes \J}$}
      \pnjcontr{1}{jb4,inc3}[cjb4]{$\J^{\bot}$}
      \pnpar{jbpjoaojobojoj,cjb4}[jbpjoaojobojojpjb]
      {$\textcolor{orange}{((\J^{\bot}\parr (\J(A)\otimes\J(B)\otimes \J))\otimes \J)\parr \J^{\bot}}$}
      \pntensor{jbpjoaojobojojpjb,j5}
      {$\textcolor{red}{(((\J^{\bot}\parr (\J(A)\otimes \J(B))\otimes \J)\otimes \J)\parr \J^{\bot})\otimes\J}$}
      \pnjaxiom{0}{j5}
    \end{proofnet}
    \caption{Gadget for a positive binary connective \(A\otimes B\)}
    \vspace{10pt}
    \label{fig:gadget-pos}
  \end{figure}
  \begin{figure}[t]
    \centering
    \vspace{\beforepn}
    \begin{proofnet}
      \pnformulae{
        ~~~~~~\pnf[a]{$A$}~\pnf[b]{$B$}\\
        ~~~~~\pnf[j1]{$\J$}~\pnf[ja]{$\vdots$}~\pnf[jb]{$\vdots$}~\pnf[j2]{$\J$}\\
        \\
        \pnf[inc]{\textcolor{blue}{$\bullet$}}~\pnf[jb1]{$\J^{\bot}$}~\pnf[jb2]{$\J^{\bot}$}\\
        \\
        \\
        \\
        ~~~~~~~~~\pnf[j4]{$\J$}~~~\pnf[jb4]{$\J^{\bot}$}~~\pnf[jb4']{$\J^{\bot}$}~~\pnf[inc2]{$\textcolor{orange}{\bullet}$}\\
        ~~~~~~~~~~~~~~~~~~~~~\pnf[inc3]{$\textcolor{orange}{\bullet}$}\\
        \\
        \pnf[j5]{$\J$}
      }
      \draw[-,dotted] (a) to (ja);
      \draw[-,dotted] (b) to (jb);
      \pntensor{j1,ja}[joa]{$\vdots$}
      \pntensor{j2,jb}[job]{$\vdots$}
      \pnpar{joa,job}[joaojob]{$\J(A)\parr\J(B)$}
      \pnjaxiom{1}{jb2,j1}
      \pnjaxiom{1}{jb1,j2}[2]
      \pnjcontr{2}{inc,jb1,jb2}[cjb]{$\J^{\bot}$}
      \pnpar{joaojob,cjb}[joaojoboj]
      {$\textcolor{blue}{\J(A)\parr\J(B)\parr \J^{\bot}}$}
      \pntensor{j4,joaojoboj}[jbpjoaojoboj]
      {$\textcolor{orange}{\J\otimes (\J(A)\parr\J(B)\parr \J^{\bot})}$}
      \pnjaxiom{2}{jb4,j4,jb4',inc2}
      \pnjcontr{1}{jb4}[cjb4]{$\J^{\bot}$}
      \pnjcontr{1}{jb4',inc3}[cjb4']{$\J^{\bot}$}
      \pnpar{jbpjoaojoboj,cjb4}[jbpjoaojobojoj]
      {$\textcolor{green}{(\J\otimes (\J(A)\parr\J(B)\parr \J^{\bot})) \parr \J^{\bot}}$}
      \pnjcontr{1}{jb4}[cjb4]{$\J^{\bot}$}
      \pntensor{jbpjoaojobojoj,j5}[jbpjoaojobojojoj]
      {$\textcolor{red}{((\J\otimes (\J(A)\parr\J(B)\parr \J^{\bot})) \parr \J^{\bot})\otimes\J}$}
      \pnjaxiom{0}{j5}
      \pnpar{jbpjoaojobojojoj,cjb4'}
      {$\textcolor{orange}{(((\J\otimes (\J(A)\parr\J(B)\parr \J^{\bot})) \parr \J^{\bot})\otimes\J)\parr \J^{\bot}}$}
    \end{proofnet}
    \caption{Gadget for a negative binary connective \(A\parr B\)}
    \label{fig:gadget-neg}
  \end{figure}  
  \begin{itemize}
    \item As before, orange links define the \emph{out- and in-synchronisation 
      points} of an edge $F$ in $R$ (for axiom links we have a unique
        out-synchronisation point placed on the positive edge, 
        it is common to the two edges in order to avoid redundancy).
  \item Blue and Green links are \emph{propagation} points, they will allow
  the incoming jumps of the premises of a link to inherit from the
  outgoing jumps of the link during cut elimination (or the outgoing jumps of the conclusion of an axiom
  link to be transfered to the positive atom after reduction).
  In particular, blue links are connected to the out-synchronisation point of the
  premises of the link being replaced.
  \item Finally, Red links are \emph{absorption} points, they will erase,
  the incoming jump of the premises of a cut after its elimination.
\end{itemize}
Given two formula occurrences $F$ and $F'$ in $R$, a \emph{jump from $F$ to $F'$} is
now encoded in $J(R)$ by connecting the out-synchronisation point of $F$ to the
in-synchronisation point of $F'$, \emph{and}, if $F'$ is the premise of a link
in $R$, to the Blue propagation point of that link. 

Let $\L$ be a set of pairs
of formula occurrences in $R$, we note $\J_{\L}(R)$ the $\MLLJ$-proof structure
corresponding to $\J(R)$ with all the jumps described in $\L$.
\end{definition}

\begin{remark}
  Although our encoding is quite heavy the resulting proof-structure has a
  number of links and edges which is linear in the number of links and edges of
  the original structure and the additional number of jumps: gadgets introduce
  at most 15 links and 22 edges for one link; and jumps introduced at most 2 new
  edges.

  As intuited in the introduction of the section and detailed below, it is during
  cut elimination that our encoding may grow the most: in order not to lose any
  ordering our encoding propagates every jumps arriving onto the premisses of a cut 
  towards edges targeted by the jumps leaving these premises after cut elimination.
  In other words, eliminating a cut in $R$ will amount to
  eliminate the corresponding gadgets in $\J(R)$ but may create $m\times n$ new
  edges where $m$ is the number of jumps arriving onto the premises of the cut and
  $n$ is the number of jumps leaving them. Hence, the number of edges in 
  $\J(R)$ after cut elimination may be an exponential of the original number of 
  jumps in the original number of cuts.

  The propagation of all jumps is sufficient in order to get
  Theorem~\ref{thm:bisim} but may generate unnecessary redundancy. The question
  of eliminating this redundancy is interesting in itself but is out of the
  scope of this paper, as this would probably involve changing the way cut
  elimination is performed and thus is left for future work.
\end{remark}

On top of reintroducing duality between positive and negative formul\ae{} with
synchronisation points, absorption and propagation points also allow to destruct or
transfer jumps during cut elimination. Indeed, one can check that:
\begin{itemize}
\item dual formul\ae{} are interpreted by dual formul\ae{};
\item one step of cut reduction
between $F$ and $F^\perp$ in $R$, written $R \to_{F,F^\perp}$,
can be simulated by several steps of cut reduction in $\J_\L(R)$:
one has to eliminate all the cuts between their corresponding gadgets and jumps.
We write $\J_\L(R) \to^*_{F,F^\perp}$ for the corresponding $\MLLJ$
proof structure.
\item during these steps of cut elimination, jumps coming to one formula are
  merged with the jumps coming out of it and its dual, on both directions. 
  Moreover premises of both formul\ae{} also inherit 
  the jump going out of these formul\ae{}, or, in the case of axiom links, the
  inheritance is put on the positive atom.
\end{itemize}

Let us make precise how $R \to_{F,F^\perp}$ and $\J_\L(R) \to^*_{F,F^\perp}$
relate. For $\L_1,\L_2$ sets of jumps (we recall that jumps are pairs of
formul\ae{}), we define:
\begin{itemize}
  \item (concatenation) $\L_1 \circ \L_2 = 
    \{(F,G)\mid \exists H,(F,H)\in \L_1, (H,G)\in \L_2 \}$
  \item (exponentiation) 
    $\L_1 * \L_2 = \pi_1(\L_1) \times \pi_2(\L_2)$, where \(\pi_i\) are projection maps.
\end{itemize}
For $\L$ a set of jumps and $F$ a formula occurrence,
we define: 
\begin{itemize}
  \item (rules as jumps)\quad 
    $\L_F = \left\{ \begin{array}{ll}
        \left\{(F_1,F),(F_2,F)\right\}& \text{ if } F =F_1 \boxempty F_2 \\
        \left\{(A,A^\perp)\right\} &\text{ if  $F\in \{A,A^\perp\}$ is an atom}
    \end{array} \right.$
  \item (in and out restrictions) \quad
    $\L_{\to \F} = \{(G,F)\in \L \mid F\in \F\}\}$,
  \qquad
 $\L_{\F\to} = \{(F,G)\in \L \mid F\in \F\}$ ;
  \item (exclusion) \qquad \ 
  $\L\backslash \F = L \backslash (\L_{\to\F} \cup \L_{\F\to})$ ;
\end{itemize}

In $\J_\L(R)$,\ the set  $I_F = \L_F \cup (\L \circ \L_F)$ corresponds to the 
jumps toward the Blue propagation point of $\J(F)$.
Cut elimination between $F$ and $F^\perp$ thus creates the following
jumps: $C_{F,F^\perp} =
(I_F \cup I_{F^\perp}) * (\L_{F\to}\cup\L_{F^\perp \to})$ (as jumps in 
  $\L_{F\to}$ are duplicated when their target is the premise of a link, 
the resulting jumps are duplicated in the same way)
and delete the one in $\L_{\to F}\cup\L_{\to F^\perp }$
($\J^\perp$ edges corresponding to $I_F \cup I_{F^\perp}$ are also
deleted but they were only replications of remaining jumps in $\L$).

Writing  $\L\!\downarrow\!\{F,F^\perp\}$ for 
$(\L\backslash \{F,F^\perp\}) \cup C_{F,F^\perp}$,
we have:
\begin{theorem}[bisimulation]
  \label{thm:bisim}
  Let \(R\) be a proof-structure,
  $\L$ be a list of pairs of formula occurrences,
  and $F, F^\perp$ be the premises of a cut in $R$,
  then
  \[\J_\L(R)\to_{F;F^\perp}^* =
  \J_{\L\downarrow\{F,F^\perp\}} (R\to_{F;F^\perp})\] 
\end{theorem}

\begin{corollary}
  Let \(R\) be a proof-structure and
  $\L$ be a list of pairs of formula occurrences, then 
  \[\seq_{\Downarrow \L} (\Downarrow R) =
  \seq(\Downarrow (\J_\L R))\] 
  where
  $\Downarrow \L$ is defined as 
  $ \L \downarrow \rho(1) \cdots\downarrow \rho(n)$
  for $\rho$ any reduction sequence from $R$ to $\Downarrow R$.
\end{corollary}
\begin{proof}
  Theorem~\ref{thm:bisim} implies 
  $\J_{\Downarrow_\rho \L} (\Downarrow_\rho R) =
  \Downarrow_\rho (\J_\L R)$ for any reduction sequence $\rho$,
  and by confluence of cut elimination in proof-structure the reduction
  path does not matter.
  
  Back in a cut-free context, one can apply 
  Theorem~\ref{thm:jumps-as-axiom} to get the result. 
\end{proof}

\paragraph{Correction}
As done for tensor in Section~\ref{sec:jumps-as-axioms}, we restrict the use
of cuts on $J$-formulae in order to have a well-defined notion of 
erasure of $\J$ from $\MLLJ$ proofs to $\MLL$ proofs:
\[
  \begin{prooftree}
    \infer0[ax$_\J(n)$]{\vdash \J, \J^\perp, \cdots,\J^\perp}
    \hypo{\vdash \Gamma, \J^\perp,\Theta_\J}
  \infer2[cut\(_\J\)]
  {\vdash \Gamma, \J^\perp, \cdots, \J^\perp, \Theta_\J}
  \end{prooftree}
\]
$\MLLJ$ is closed under cut elimination and following the 
same reasoning as before we still have
\begin{theorem}
  Let \(R\) be a proof-structure, and $\L$ be a set of pairs of 
  formula occurrences,   then
  \[\erasej{\seq\ (\J_{\L}(R))} =  \seq_{\L}\ (R).\]
\end{theorem}
So as before we have an encoding ($\J_{\L}$) of graphical jumps in $\MLLJ$  (and its corresponding decoding $\mathsf{erase}_\J$) such that:
  \begin{center}
    \begin{tikzcd}
      \PSJ \ar[r,"\seq"] & \mathcal{P}(\MLLJ) \ar[d,"\mathcal{P}(\erasej{\_})"]\\
      \PS  \ar[u,"\J_{\L}(\_)"] \ar[r,"\seq_{\L}"]& \mathcal{P}(\MLL)
    \end{tikzcd}
  \end{center}
  and, furthermore, by Theorem~\ref{thm:bisim},  there exists a {bisimulation} procedure so
  that the above diagram is stable under cut-elimination.

\section{Encoding in MELL + Mix}
\label{sec:implem}

In this section, in order to finish our logical internalization of jumps, we
eventually show how jumps, as defined in Figure~\ref{sec:jumps-as-axioms}
and~\ref{sec:cuts}, can be implemented in the exponential fragment of $\MELL$ +
$\Mix$.

As explained in Section~\ref{sec:jumps-as-axioms}, we need $\J$ and $\J^\perp$ 
to be able to contract. This can be achieved using the contraction rule of the
$\wn$-exponential in $\MELL$, so our encoding of $\J$ and $\J^\perp$ must involve
this exponential modality.
$\J$ only needs the ability to contract at the level
of axioms while $\J^\perp$ is free to contract everywhere, furthermore,
$\J$ and $\J^\perp$ must be dual from each other, as such, we set:
\[\J = \oc\wn\axj \qquad \qquad \J^\perp = \wn\oc\axj^\perp\]
where $\axj$ is a special atom that does not appear in regular
$\MLL$ formulae.
\begin{figure}[t]
  \centering
  \vspace{\beforepn}
  \subcaptionbox{\(\J(0)\)\label{fig:encoding-j0}}{
    \pnet{
      \pnformulae{
        \pnf[?aj]{$\wn\axj$}
      }
      \pninitial[?]{$\wn$}{?aj}
      \pnbox{?aj,?}
      \pnprom{?aj}{$\oc\wn\axj$}
    }
  }\hfill
  \subcaptionbox{\(\J(n), n>0\)\label{fig:encoding-jn}}{
    \pnet{
      \pnformulae{
        \pnf[jb1]{$\axj$}~\pnf[j1]{$\axj^{\bot}$}~~~\pnf{\textellipsis}~~
        \pnf[jbn]{$\axj$}~\pnf[jn]{$\axj^{\bot}$}\\
      }
      \pnaxiom[ax1]{j1,jb1}
      \pnbox{jb1,j1,ax1}
      \pnprom{j1}[!j1]{$\oc\axj^{\bot}$}
      \pnauxprom{jb1}[?jb1]{$\wn\axj$}
      \pnaxiom[axn]{jn,jbn}
      \pnbox{jbn,jn,axn}
      \pnprom{jn}[!jn]{$\oc\axj^{\bot}$}
      \pnauxprom{jbn}[?jbn]{$\wn\axj$}
      \pnexp{}{?jb1,?jbn}[?jb]{$\wn\axj$}[1][-2cm]
      \pnbox{?jb,ax1,axn,jb1,jn}
      \pnprom{?jb}{$\oc\wn\axj$}
      \pnauxprom{!j1}{$\wn\oc\axj^{\bot}$}
      \pnauxprom{!jn}{$\wn\oc\axj^{\bot}$}
    }
  }\hfill
  \subcaptionbox{\(\J^{\bot}(n), n > 0\)\label{fig:encoding-jbn}}{
    \pnet{
      \pnformulae{
        \pnf[?1]{$\wn\oc\axj^{\bot}$}~~\pnf{\textellipsis}~~\pnf[?n]{$\wn\oc\axj^{\bot}$}
      }
      \pnexp{}{?1,?n}{$\wn\oc\axj^{\bot}$}
    }
  }
  \caption{Encoding of the combinators on $\J$ and $\J^{\bot}$}
  \label{fig:encoding}
\end{figure}

In Figure~\ref{fig:encoding} we show how the static properties of jumps (that are $\J(n\geq 0)$ and $\J^\perp(m>0)$) can be derived via this encoding. 
The dynamics properties of jumps given described in Section~\ref{eq:ctr}  
can then be derived from the usual dynamics of exponential in $\MELL$.

\section{Conclusion and future work}
\label{sec:conclusion}

We have presented a logical framework, $\MLLJ$, in which
$\MLL$-proof-structures with jumps can be encoded.
This framework can be seen both as an extension of $\MLL$
  with a special atom and as a fragment of $\MELL$ + $\Mix$.
  As such, correctness criteria for proof-structures in these logics induce correctness criteria 
  for $\MLLJ$-proof-structures. In particular, the Danos-Regnier correctness criterion for $\MLL$ induces
  correctness criterion for $\MLLJ$, which corresponds to the correctness criterion
  given in~\cite{DiGiamberardinoFaggian:conf}, in a setting where jumps in proof structures are treated from a purely graph-theoretical point of view.
 
  Our encoding is compatible with the dynamics of proof-nets: cut elimination in our framework
  achieves the graphical dynamics of jumps considered 
  in~\cite{DiGiamberardinoFaggian:conf,giamberardino_2018}.
  On top of being the most natural one (it makes two jumps to
  compose if one targets the premise of a link involved in a cut 
and the other leaves one of these links), we believe that this
dynamics is justified by the following conjecture:
\begin{conjecture}
  Let $R$ be a $\MLL$ proof-net, and $\L$ be a set of pairs of formula
  occurrences such that $\sec_\L(R) = \{\pi\}$, i.e. $\L$ adds enough 
  constraints on $R$ for its sequentialization set to be a singleton, then
  \[\sec(\Downarrow \J_\L(R)) = 
  \{\pi' \mid \exists \rho, \pi' = \Downarrow_\rho \pi\}\]
\end{conjecture}
In other words, the freedom of sequentialization induced by the cut 
elimination procedure in Jumped proof structures corresponds exactly 
to the choices made during cut elimination in the sequent calculus.
$\supseteq$ is a consequence of Theorem~\ref{thm:bisim}, $\subseteq$
however is left as future work.

Our framework draws a clear line between formul\ae{} encoding jumps and
regular $\MLL$ formul\ae{}. This advocates for easy extensions to other
logics such as $\MALL$ or $\MELL$ with or without the $\Mix$ rule.
In particular we plan, in expanding this work to $\MELL$ to study in detail the correspondence between the dynamics that is induced by our modelling of jumps and the cut-elimination considered in~\cite{giamberardino_2018} for the exponentials as well as how the notion of cones introduced by di Giamberardino as a replacement and generalization of exponential boxes translates in our setting.


This work fits in a general tendency to try to make sense logically of the
graph-theoretic foundations of proof-nets. Here, we reinterpreted jumps as a
purely logical structure.

  \paragraph{Acknowledgments} The authors thank the anonymous reviewers for
  their remarks and advice, and Paolo Pistone for pointing to them the
  connection to Hughes' work.

  Aurore Alcolei was supported by the DIM RFSI.

  Alexis Saurin was partially supported by the ANR project ReCiProg.

\bibliographystyle{entics}
\bibliography{biblio}

\appendix

\end{document}